\documentclass[groupedaddress]{revtex4}

\usepackage[utf8]{inputenc}
\usepackage{amsmath,color,verbatim,graphicx,amsfonts,amssymb}

\newtheorem{definition}{Definition}

\newtheorem{corollary}{Corollary}
\newtheorem{proposition}{Proposition}

\newenvironment{proof}[1][Proof]{\noindent\textbf{#1.} }{\ \rule{0.5em}{0.5em}}

\newcommand{\acco}[1]{\left\{#1\right\}}

\newcommand{\pa}[1]{\left(#1\right)}

\newcommand{\BN}{\mathcal{BN}}
\newcommand{\TN}{\widetilde{\mathcal{N}}}

\newcommand{\incl}{\subseteq}

\newcommand{\C}{\mathcal{C}}

\newcommand{\N}{\mathcal{N}}

\newcommand{\Z}{\mathbb{Z}}

\newcommand{\id}{\operatorname{id}}

\newcommand{\Swap}[1]{S_{#1}}
\newcommand{\Loc}{\operatorname{Loc}}

\begin{document}

\title[Block representation of RCA, and TSCA]{A simple block representation of reversible cellular automata with time-symmetry}	

\author{Pablo Arrighi}
\affiliation{Universit\'e de Grenoble, LIG, 220 rue de la chimie, 38400 Saint-Martin-d'H\`eres, France\\
and \'Ecole Normale Supérieure de Lyon, LIP, 46 Allée d'Italie, 69364 Lyon, France}
\email{Pablo.Arrighi@ens-lyon.fr}

\author{Vincent Nesme}
\affiliation{QMIO, Freie Universität Berlin, Arnimallee 14, 14195 Berlin, Germany}
\email{Vincent.Nesme@qipc.org}


\begin{abstract}
Reversible Cellular Automata (RCA) are a physics-like model of computation consisting of an array of identical cells, evolving in discrete time steps by iterating a global evolution $G$. Further, $G$ is required to be shift-invariant (it acts the same everywhere), causal (information cannot be transmitted faster than some fixed number of cells per time step), and reversible (it has an inverse which verifies the same requirements). An important, though only recently studied special case is that of Time-symmetric Cellular Automata (TSCA), for which $G$ and its inverse are related via a local operation. In this note we revisit the question of the Block representation of RCA, i.e. we provide a very simple proof of the existence of a reversible circuit description implementing $G$. This operational, bottom-up description of $G$ turns out to be time-symmetric, suggesting interesting connections with TSCA. Indeed we prove, using a similar technique, that a wide class of them admit an Exact block representation (EBR), i.e. one which does not increase the state space. 
\end{abstract}

\maketitle

\section*{Introduction}

\noindent {\em RCA, Block representation.} In \cite{Kari_blocks}, Kari showed that any one-dimensional or two-dimensional reversible cellular automaton (RCA) can be expressed as a composition of finite reversible gates (or `block permutations') and partial shifts. In two dimensions the proof is quite involved, the representation requires three layers of blocks, and it has been proved that this cannot be brought down to a two-layered block representation \cite{Kari_circuit_depth}; The problem is still open in higher dimensions.

However we may not need an exact representation, and be willing to encode our original cells into some larger ones (or equivalently to interleave some ancillary cells), as proposed in \cite{durandlose}. Then the construction of \cite{Kari_circuit_depth} shows that even $n$-dimensional RCA admit a two-layered block representation. In some sense what we are doing then is simulating the original RCA in a way which preserves the spatial layout of cells, with another, simpler RCA that we know admits a two-layered block representation. In this sense the intrinsically universal RCA \cite{Durand-LoseLATIN} also accomplishes this task.\\
{\em Our Section \ref{sec:SBR} revisits this issue in a minimalistic manner: In our construction each block can be interpreted a reversible version of the local update rule of the CA, moreover its size turns out to be exactly that of the Block Neighborhood introduced in \cite{block}.}

\noindent {\em TSCA, EBRs.} Recently another line of investigation has emerged which refines the now well-studied concept of RCA to admit a further requirement: That of time symmetry.  In simple terms, a CA $G$ is time-symmetric if $G$ is its own inverse up to a simple recoding $H$ of the cells.  More formally, $G^{-1}=H G H$ with $H$ a self-inverse CA. Credit must be given to \cite{gamo} for emphasizing time-symmetry as a property of CA, which has barely been studied for its own sake thus far.  It is clear nevertheless that many instances of time-symmetric CA (TSCA) can be encountered in the literature, as discussed in \cite{gamo} (for instance the Margolus lattice gas model).  In the above-discussed non-exact Block representation of RCA \cite{Kari_circuit_depth} just like in ours, the author first encodes a RCA $F$ into a TSCA $G_F$, and then provides an EBR of $G_F$.  As a consequence, one may wonder whether these issues, block representations of RCA and TSCA are only accidentally related, or whether exhibiting a reversible local implementation mechanism for $G$ amounts to unravelling the time-symmetry of $G$.\\
{\em Our Section \ref{sec:TSCABR} begins to explore this issue by showing the existence of an EBR for squares of locally time-symmetric CA. }

\section{A simple block representation}\label{sec:SBR}

In the classical picture a CA $G$ is usually defined by a local update rule $\delta$, namely a function from $\Sigma^\N$ to $\Sigma$, giving the new state of a cell as a function of the old state of its neighbours; It can be thought as a `local mechanism' for implementing $G$.
In other words, $\delta$ can be viewed as a local gate, and $G$ a circuit made by infinitely repeating $\delta$ across space as in Fig. \ref{fig:classicalcircuit}.

\begin{figure}[htbp]
\centering
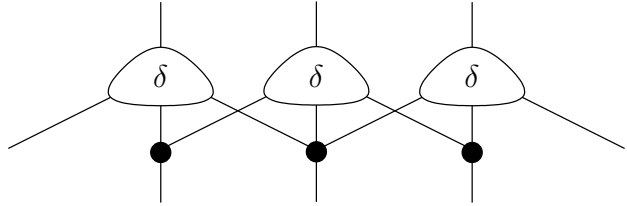
\caption{The trivial circuit representation of a classical CA from its local update rule.\label{fig:classicalcircuit}}
\end{figure}

Using a local update rule to define RCA is of course possible, but for a circuit representation of $G$ one may wish to use a local mechanism that is itself reversible --- for instance in the context of quantum mechanical devices or due to Landauer's principle.  And indeed it is the case that every RCA $G$ admits a reversible circuit implementation.  Proving the existence of such reversible circuits is the business of the aforementioned block representation theorems for RCA. It could be regretted, however, that in these theorems the reversible local gates (a.k.a blocks) constitutive of the reversible circuits (a.k.a block representations) end up looking quite different from $\delta$.  I.e. they are hard to interpret as reversible versions of the local update rule.

The following proof of the block representation theorem for RCA is hopefully simpler to understand. It starts off by defining a reversible update operator $K_0$, which can be interpreted as a reversible version of the local update rule $\delta$.  We will define $K_0$ globally, in a way that does not make it obvious that it is actually a block permutation --- but we will then proceed to show that it is the case.  Notice that it is impossible to implement CA of non-trivial Welch index --- for a definition, cf. section~3 of \cite{Kari_blocks} --- without shifts or auxiliary space: In our case, we use auxiliary space, which results in the collateral damage of implementing, in parallel to $G$, its inverse on the auxiliary strip.

Repeatedly we will define a bijection $f$ from a set of words written on some fixed set of cells $X$, and then wonder whether $f$ could be defined on a smaller subset.  We will say that $f$ is localized upon $Y\incl X$ if we can write $f=f_Y\times\id_{Y\setminus X}$, i.e. if $Y\setminus X$ is superfluous in the definition of $f$.  For instance, a bijection of $\Sigma^\Z$ that applies a permutation of the alphabet on cell $0$ and leaves the other cells untouched is localized upon $Y\incl \Z$ if $Y$ contains $0$; The identity is localized on the empty set.

From the definition, it is obvious that if $f$ is localized upon $Y$ and $Y\incl Z\incl X$, then $f$ is also localized upon $Z$.  Slightly less trivial is the property that, whenever $f$ is localized upon $Y$ and $Z$, then it is also localized upon their intersection $Y\cap Z$.  From there follows the existence of the smallest $Y$ upon which $f$ is localized, which is called the \emph{localization} of $f$, and denoted $\Loc(f)$.  So, back to our elementary example where $f$ is a permutation $\pi$ of $\Sigma$ applied solely on cell $0$, $\Loc(f)=\left\{\begin{array}{ll}\emptyset&\text{if $\pi=\id$} \\ \acco{0}&\text{otherwise} \end{array} \right.$.

In general, $K_0$ is not localized upon the neighborhood of $G$.  We will show however that its localization is $\BN$, the Block neighborhood defined in \cite{block} whose definition we will recall. 
Hence it can thus be viewed as  a block permutation of size $|\BN|$.  
The last step of the proof is just to show that $G$ a circuit made by infinitely repeating $K$ across space.

\subsection*{Reversible updates $K_i$\ldots}
\label{sec:SimpleBR}

In the classical picture, the local update rule $\delta$ looks at a neighborhood $\cdots c_{-1}c_0c_{1}\cdots $ and computes $G(c)_0$, but it leaves all the other cells uncomputed. Can we, in a similar fashion, define a reversible update $K_0$ which focuses on computing $G(c)_0$? Moreover can we, in an again a similar fashion, define it solely in terms of $G$? A naive, operational approach would be to: 1. Apply $G$. 2. Swap $G(c)_0$ out of the system. 3. Apply $G^{-1}$. This will turn out to work. Technically, we will extend the alphabet to $\Sigma^2$.  For $i$ running over all cells, we denote by $\Swap{i}$ the swap acting only on position $i$ according to $\pa{\begin{array}{rcl}\Sigma^2&\to&\Sigma^2 \\ (a,b)&\mapsto &(b,a)\\ \end{array}}$.

\begin{definition}[reversible update]\label{def:K}
The {\em reversible update} $K_i$ is the function from $\C_{\Sigma^2}\simeq \C_{\Sigma}^2$ to itself given by the following composition
$$K_i=(G^{-1}\times\id) \Swap{i} (G\times\id)$$
where $\C_\Sigma$ denotes the space of configurations of cells having alphabet $\Sigma$.
\end{definition}

We can right now formulate the important remark that the $K_i$-s commute.  We will later prove with Proposition~\ref{prop:Kaloc} that each $K_i$, despite being defined globally, is actually a local permutation, acting in some neighborhood of cell $i$; Let us admit this fact within this paragraph.  With these informations in mind, it makes sense to define the infinite product $\prod\limits_{i} K_i$.  Indeed, for any given cell, the number of $K_i$-s acting on this cell is finite; Therefore the composition of all the $K_i$-s can be written as a circuit of finite depth and is thus perfectly well-defined.  Moreover, it is equal to $(G^{-1}\times \id)\Swap{}(G\times\id)$, where $\Swap{}=\prod\limits_{i} \Swap{i}$. Therefore we have $\Swap{}\prod\limits_{i} K_i=G\times G^{-1}$. 

Let us take a closer at $K_0$.  Start with a configuration $\ldots(c_i,d_i)\ldots$.  Applying $G\times\id$ takes it to $\ldots(G(c)_i,d_i)\ldots$.  Then $\Swap{0}$ turns it into 

$$\ldots(G(c)_{-2},d_{-2}),(G(c)_{-1},d_{-1}),(d_{0},G(c)_{0}),(G(c)_{1},d_{1}),(G(c)_{2},d_{2})\ldots$$  

So $K_0$ leaves the second component unchanged, except in position $0$.  In fact, the rest of the second component could be left out in the definition of $K_0$, since it plays no role.  Specifically, one can write $K_0$ as a product of the identity on these cells and of some bijection of $\C_{\Sigma}\times\Sigma$.  The left component, after applying $K_0$, finds itself in the state $G^{-1}(\ldots G(c)_{-2}G(c)_{-1}d_0G(c)_{1}G(c)_{2}\cdots)$.  Of course, outside of some neighborhood of $0$, this is the identity; But that triviality alone is not enough to conclude that $K_0$ is localized upon a finite number of cells.  We are going to check that it is indeed the case, and moreover that its localization is a rather remarkable set.

\subsection*{\ldots are localized within the Block Neighborhood $\BN$ \ldots}

\begin{figure}[htbp]
\centering
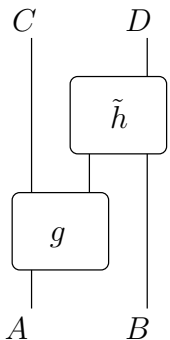
\caption{Semilocalizability.\label{fig:defsemiloc}}
\end{figure}

\begin{figure}[htbp]
\centering
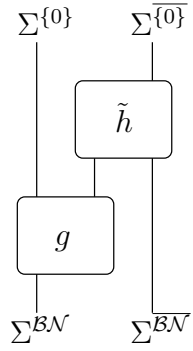
\caption{The block neighborhood.\label{fig:semiloc}}
\end{figure}

In~\cite{block}, the authors introduced the block neighborhood $\BN$ of a RCA, using the concept of semilocalizability that appeared in \cite{semicausal} in the context of quantum information theory.  Given a bijection $F:X\to Y$ and a decomposition of $X$ and $Y$ in respectively $A\times B$ and $C\times D$, $F$ is said to be semilocalizable (with respect to this decomposition) when it can be written in the form of Figure~\ref{fig:defsemiloc}, where $g$ and $\tilde h$ are themselves bijections.  The quantum neighborhood of a RCA $F$ is then the smallest subset $\BN$ such that, as a function from $\Sigma^{\BN}\times\Sigma^{\overline{\BN}}$ to $\Sigma^{\acco{0}}\times\Sigma^{\overline{\acco{0}}}$, $F$ is semilocalizable --- see Figure~\ref{fig:semiloc} for an illustration. 

The definition of the block neighborhood was motivated by the fact that it is both the (quantum) neighborhood of the quantum CA obtained by linearization from a RCA, and obviously related to the decomposition of a QCA into a product of local permutations, a link that we make more precise in this article.  More details on $\BN$ are to be found in \cite{block}, where it is the object of definition~1.9, and where explicit bounds on $\BN$ are given in function of the neighborhoods of $G$ and of its inverse.  We will not need these bounds, except for the fact that they do prove that $\BN$ is finite:
\begin{itemize}
\item $\BN$ is included in $(\N - \N + \TN) \cap (\TN - \TN + \N)$, with $\TN$ the transpose of the inverse neighborhood $\N^{-1}$. There are examples saturating this bound;
\item $\BN(G^k)/k$ tends towards $\N(G^k)\cup\TN(G^k)$ in the limit where $k$ goes to infinity, with $\BN(G^k)$ the Block Neighborhood of $G^k$ etc.
\end{itemize}
In the definition of $K_0$, cells are divided into two subcells, so that these subcells are naturally indexed by $\acco{0,1}\times\Z$.  We now prove that the localization of $K_0$ is essentially the block neighborhood $\BN$; As $\BN$ is also the \emph{quantum} neighborhood, i.e. the neighborhood when inputs are not just words but can be linear combination on words (cf.~\cite{block}), this gives a nice way to characterize the quantum dynamics in a purely classical setting.

\begin{proposition}\label{prop:Kaloc}
Consider a RCA $G$, and let $K_0$ be its reversible update. Then $\Loc(K_0)=\acco{0}\times\BN\cup\acco{(1,0)}$. 
\end{proposition}

\begin{proof}
$[\subseteq]$. Consider a  $\tilde hg$-decomposition of $G$ in the manner of Figure~\ref{fig:semiloc}. Then $g$ is localized upon $\BN$, $\tilde h$ outside of cell $0$, and
\begin{eqnarray*}
K_0&=&(G^{-1}\times\id) \Swap{0} (G\times\id)\\
&=&((\tilde hg)^{-1}\times\id) \Swap{0} ((\tilde hg)\times\id)\\
&=&(g^{-1}\times\id)(\tilde h^{-1}\times\id) \Swap{0} (\tilde h\times\id) (g\times\id)\\
K_0&=&(g^{-1}\times\id) \Swap{0} (g\times\id)
\end{eqnarray*}
where the last line follows from the fact that $\Loc(\tilde h)$ does not contain $\acco{0}$, whereas $\Swap{0}$ is  localized upon cell $0$.  From this last line we can read $\Loc(K_0)\incl\acco{0}\times\BN\cup\acco{(1,0)}$.\\
$[\supseteq]$. {\em Note that this second inclusion is no needed for the proof of the Block representation; It is provided here just for completeness.}  As we have already mentioned,  $\Loc(K_0)$ is of the form $\Loc(K_0)_0\cup \acco{(1,0)}$.  So $\Loc{\prod\limits_{n\neq 0} K_n}$ does not contain $(1,0)$.  But $K_0\prod\limits_{n\neq 0} K_n= (G^{-1}\times \id)\Swap{}(G\times\id)$.  For $a\in\Sigma$, let $X_a$ be the subset of words on $\Loc(K_0)$ that are equal to $a$ on $(1,0)$.  The image of $X_a$ by $S_0(G\times\id)$ is of the form $Y_a\times\Sigma$, where $Y_a$ is the set of words on $\Loc(K_0)_0\cup\acco{(0,0)}$ that are equal to $a$ in $(0,0)$, and $\Sigma$ is localized on $(1,0)$.  Therefore the image of $X_a$ by $K_0$ is also of the form $Z_a\times\Sigma$ for some subset $Z_a$ of the words on $\Loc(K_0)_0$. \\ 
Furthermore, we know that there exists a bijection finishing the job after the isolation of $G(c)_0$ by $K_0$, namely $\prod\limits_{n\neq 0}K_n$.  We must thus have a semilocalization of $G$ with respect to $\Loc(K_0)_0$: In figure~\ref{fig:semiloc}, $K_0$ plays the role of $g$, $\BN$ is $\Loc(K_0)_0$, and $\tilde h$ is $\prod\limits_{n\neq 0}K_n$.  Since $\BN$ is the smallest set fulfilling this property, it must then be included in $\Loc(K_0)_0$.
\end{proof}

\subsection*{\ldots and thus implement $G$.}

Combining the above results we obtain the following:
\begin{corollary}[$G\times G^{-1}=\Swap{}(\prod K)$]\label{prop:produit}
Consider a RCA $G$, and let $K$ be its reversible update. Consider the function $G\times G^{-1}$ from $\C_{\Sigma}^2$ to $\C_{\Sigma}^2$. We have that 
$$G\times G^{-1}=\Swap{}\prod_i K_i\quad\textrm{with}\quad\Loc(K_0)=\acco{0}\times\BN\cup\acco{(1,0)}.$$
\end{corollary}

Hence we have here a proof that all RCA admit a block representation, the third of its genre \cite{Kari_blocks,durandlose}, but hopefully also the most straightforward, as it simply takes the form a product of reversible updates.  There is one bad and one good news about this proof. The bad news is that it provides only a non-exact Block representation of RCA, leaving it open whether $n>2$-dimensional RCA admit an EBR or not.  The good news is that it provides an EBR for those TSCA which are of the form $G\times G^{-1}$.  This suggests that we should look at the relation between EBRs and time-symmetry of CA.

\section{EBRs and time-symmetry} \label{sec:TSCABR}

The core of the argument that we developed in the previous section for the existence of an EBR for $G\times G^{-1}$ could be restated as follows:  Say $F$ and $H$ are RCA such that $H$ admits an EBR, then so does $FHF^{-1}$! Indeed, if $H=\prod\limits_{i}B_i$, 
then $FHF^{-1}=\prod\limits_{i}FB_i F^{-1}$.  Moreover following Proposition \ref{prop:Kaloc}.$[\subseteq]$, the blocks $FB_i F^{-1}$ are localized, at most, on the localization of $B_i$ extended by $\BN(F)$ the block neighborhood of $F$; Hence each of them is finitely localized, i.e. is itself a block permutation.\\
In Section~\ref{sec:SimpleBR} we applied this argument with $F=G^{-1}\times \id$ and $H=\Swap{}$, which admits a trivial block representation $\Swap{}=\prod\limits_{n\in\Z}\Swap{n}$.  This gave an EBR of $(G^{-1}\times \id)\Swap{}(G\times \id)$, which is only a swap away from $G\times G^{-1}$.  
In fewer words, $G\times G^{-1}$ admits an EBR because the set of RCA having this property 
\begin{itemize}
\item contains the permutations of $\Sigma$, and
\item is a normal subgroup of the group of RCA.
\end{itemize}

Having generalized this procedure, let us now have a look at what it tells us in the context of TSCA.

\begin{definition}[Locally Time-Symmetric CA]\label{def:LTSCA}
A RCA $G$ is a {\em locally time-symmetric CA} (LTSCA) if there exists an involution $h$ of $\Sigma$ such that $G^{-1}=H G H$, with $H=\prod_i h$.
\end{definition}
Our definition of LTSCA is identical to that of TSCA given in \cite{gamo} except for one extra condition: We further demand that the RCA $H$ be of radius zero.  On this question of the locality of $H$, let us quote the authors of this first paper introducing TSCA \cite{gamo}: ``\emph{Requiring $H$ to be a CA is somewhat arbitrary, [\ldots] the reason for this restriction is that we expect reversibility (including the particular case of time-symmetry) to be a local property.}''. Moreover, whilst the theoretical results they prove are valid for $H$ an involution RCA of arbitrary radius, it also true that in all of the examples provided, $H$ is of radius zero. In fact, one may wonder whether there LTSCA and TSCA are not equivalent up to a simple encoding.\\
Anyhow, if $H$ has radius zero, then in particular it admits an EBR, and so does $GHG^{-1}H=G^2$.  Therefore, the squares of LTSCA have EBRs:

\begin{corollary}[EBR of ${\textrm{LTSCA}}^2$]\label{prop:BRofLTSCA}
Let $G$ be an LTSCA with respect to an involution $h$ 
. We have $G^2= H \prod\limits_i L_i$, where $L_i=G^{-1} h_i G$, furthermore $\Loc(B_0)\subseteq\BN$.
\end{corollary}

Some remarks are in order:
\begin{itemize}
\item $h_0$ plays the role that $\Swap{0}$ had in section \ref{sec:SimpleBR}.  Likewise, in the standard examples of TSCA \cite{gamo}, $H$ can be interpreted as a swap.  This is certainly the case in particular for the standard time-symmetrizations $G\times G^{-1}$ of any RCA $G$, as in Prop. 5.3. of \cite{gamo}.
\item This time the block representation is an exact one, hence it is remarkable that LTSCA have this property given the difficulty of finding the EBRs of $n>2$-dimensional RCA. Nevertheless, the representation applies to $G^2$ and not $G$ itself.  Simply proving that any involutive RCA admits an EBR is probably difficult, as it gets dangerously close to solving the aforementioned open problem.
\end{itemize}

\section*{Conclusion}

\noindent {\em Generalizations.} As in \cite{block}, the block representation defined in Section~\ref{sec:SimpleBR}, and the proof that it is of minimal size, rely only on notions on neighborhood, while others characteristics of CA, such as finiteness of the alphabet and translation invariance, are simply irrelevant. Moreover, whilst the arguments we have provided in this paper are purely classical, they have their counterparts in the field of quantum CA \cite{Schumacher}, some of which were of direct inspirations to this paper \cite{ANW3}. Part of our motivation was to make these techniques available to classical CS. 

\noindent {\em Questions, answers and  more questions.} Why is time-symmetry such a key step Block representations of RCA? In this paper gave a simple proof of the block representation of RCA, which partly explains this role. Could it be that TSCA admit an EBR? In this paper we gave a simple proof of the EBR of squares of LTSCA. These are all but partial answers, suggesting that many questions remain on the topic of understanding differences in structure between RCA and TSCA, TSCA and LTSCA. There might lie a path towards EBRs of RCA in arbitrary dimensions.

\acknowledgements
\label{sec:ack}
The authors would like to thank Jarkko Kari, Anah\'i Gajardo, and funding by the Deutsche Forschungsgemeinschaft (Forschergruppe 635) and ANR JJC CCausaQ.

\bibliographystyle{alpha}
\bibliography{biblio}

\begin{thebibliography}{ANW11}

\bibitem[AN10]{block}
Pablo Arrighi and Vincent Nesme.
\newblock {The Block Neighborhood}.
\newblock In TUCS, editor, {\em {Proceedings of JAC 2010}}, pages 43--53,
  Turku, Finlande, December 2010.

\bibitem[ANW11]{ANW3}
Pablo Arrighi, Vincent Nesme, and Reinhard~F. Werner.
\newblock Unitarity plus causality implies localizability.
\newblock {\em Journal of Computer and System Sciences}, 77(2):372--378, March
  2011.

\bibitem[DL95]{Durand-LoseLATIN}
J\'{e}r\^{o}me Durand-Lose.
\newblock Reversible cellular automaton able to simulate any other reversible
  one using partitioning automata.
\newblock In {\em Proceedings of the Second Latin American Symposium on
  Theoretical Informatics}, LATIN '95, pages 230--244, London, UK, 1995.
  Springer-Verlag.

\bibitem[DL01]{durandlose}
J{\'e}r{\^o}me Durand-Lose.
\newblock Representing reversible cellular automata with reversible block
  cellular automata.
\newblock In Robert Cori, Jacques Mazoyer, Michel Morvan, and R\'emy Mosseri,
  editors, {\em Discrete Models: Combinatorics, Computation, and Geometry,
  DM-CCG~'01}, volume~AA of {\em Discrete Mathematics and Theoretical Computer
  Science Proceedings}, pages 145--154, 2001.

\bibitem[ESW02]{semicausal}
T.~Eggeling, Dirk Schlingemann, and Reinhard~F. Werner.
\newblock Semilocal operations are semilocalizable.
\newblock {\em Europhysics Letters}, 57(6):782--788, 2002.

\bibitem[Kar96]{Kari_blocks}
Jarkko Kari.
\newblock Representation of reversible cellular automata with block
  permutations.
\newblock {\em Mathematical Systems Theory}, 29(1):47--61, 1996.

\bibitem[Kar99]{Kari_circuit_depth}
Jarkko Kari.
\newblock On the circuit depth of structurally reversible cellular automata.
\newblock {\em Fundam. Inf.}, 38(1-2):93--107, 1999.

\bibitem[MG10]{gamo}
Andr{\'e}s Moreira and Anah{\'\i} Gajardo.
\newblock {Time-symmetric Cellular Automata}.
\newblock In TUCS, editor, {\em {Proceedings of JAC 2010}}, pages 180--190,
  Turku, Finlande, December 2010.

\bibitem[SW04]{Schumacher}
Benjamin Schumacher and Reinhard~F. Werner.
\newblock Reversible quantum cellular automata.
\newblock \url{arXiv:quant-ph/0405174}, May 2004.

\end{thebibliography}
\label{sec:biblio}

\end{document}